\documentclass[12pt]{amsart}
\usepackage{amssymb,amscd}
\usepackage{verbatim}

\usepackage{amsmath,amssymb,graphicx,mathrsfs} % my new
\usepackage[colorlinks=true,allcolors = blue]{hyperref} % my new

\let\frak\mathfrak

\def\>{\relax\ifmmode\mskip.666667\thinmuskip\relax\else\kern.111111em\fi}
\def\<{\relax\ifmmode\mskip-.333333\thinmuskip\relax\else\kern-.0555556em\fi}
\def\vsk#1>{\vskip#1\baselineskip}
\def\vv#1>{\vadjust{\vsk#1>}\ignorespaces}
\def\vvn#1>{\vadjust{\nobreak\vsk#1>\nobreak}\ignorespaces}

 \let\tsize\textstyle

\newtheorem{thm}{Theorem}[section]
\newtheorem{cor}[thm]{Corollary}
\newtheorem{lem}[thm]{Lemma}

\theoremstyle{definition} % My June 14 2017
\newtheorem{exmp}[thm]{Example}

\numberwithin{equation}{section}

\theoremstyle{definition}
\newtheorem*{rem}{Remark}

\let\mc\mathcal
\let\nc\newcommand

\let\ka\kappa
\let\la\lambda
\let\La\Lambda

\let\phi\varphi
\let\si\sigma

\let\Om\Omega

\let\der\partial

\let\ox\otimes

\let\ge\geqslant
\let\geq\geqslant
\let\le\leqslant
\let\leq\leqslant

\let\on\operatorname
\let\bi\bibitem
\let\bs\boldsymbol

\def\C{{\mathbb C}}
\def\Z{{\mathbb Z}}

\def\F{{\mathbb F}}

\def\End{\on{End}}

\def\beq{\begin{equation}}
\def\eeq{\end{equation}}
\def\be{\begin{equation*}}
\def\ee{\end{equation*}}

\nc{\bea}{\begin{eqnarray*}}
\nc{\eea}{\end{eqnarray*}}
\nc{\bean}{\begin{eqnarray}}
\nc{\eean}{\end{eqnarray}}
\nc{\Il}{{\mc I_{\bs\la}}}
\nc{\bla}{{\bs\la}}
\nc{\Fla}{\F_\bla}
\nc{\tfl}{{T^*\Fla}}
\nc{\GL}{{GL_n(\C)}}
\nc{\GLC}{{GL_n(\C)\times\C^*}}

\def\KZ/{{\slshape KZ\/}}
\def\qKZ/{{\slshape qKZ\/}}
\def\XXX/{{\slshape XXX\/}}

\def\Sing{{\on{Sing}}}
\def\sll{{\frak{sl}}}

\def\slt{{\frak{sl}_2}}

\def\K{{\mathbb K}}
\def\A{{\mathbb A}}

\def\Sym{{\on{Sym}}}

\begin{document}

\hrule width0pt
\vsk->

\title[Eigenvectors of $p$-Curvature for KZ--type Differential Connections]
{Eigenvectors and Eigenvalues of $p$-Curvature
\\
Operators for
KZ--type Differential Connections}

\author[Vitaly Tarasov and Alexander Varchenko]
{Vitaly Tarasov$\>^\circ$ and Alexander Varchenko$\>^\star$}

\begin{abstract}

A KZ-type connection in characteristic $p$ has $p$-curvature operators. These operators are commuting endomorphisms of the connection. We present a Bethe-ansatz-type construction of their eigenvectors and eigenvalues. The eigenvector is constructed as the value of a characteristic $p$ hypergeometric integral and extends to a flat eigensection
over the Frobenius neighborhood of the corresponding Bethe point.

\end{abstract}

\maketitle

\begin{center}
{\it $^\circ$Department of Mathematical Sciences,
Indiana University,
402 North Blackford St,
\\
Indianapolis, IN 46202-3216, USA\/}
\vsk.5>
{\it $^{\star}\<$Department of Mathematics, University
of North Carolina at Chapel Hill\\ Chapel Hill, NC 27599-3250, USA\/}
\end{center}

{\let\thefootnote\relax
\footnotetext{\vsk-.8>\noindent
$^\circ\<${\sl E\>-mail}:\enspace vtarasov@iu.edu\>,
supported in part by Simons Foundation grants \rlap{430235, 852996}
\\
$^\star\<${\sl E\>-mail}:\enspace anv@email.unc.edu\>,
supported in part by Simons Foundation grant TSM-00012774}}

\vsk>
{\leftskip3pc \rightskip\leftskip \parindent0pt \Small
{\it Key words\/}:
KZ connection, integral representation, commuting $p$-curvature operators, eigenvectors and eigenvalues
\vsk.6>
{\it 2020 Mathematics Subject Classification\/}: 81R50, 33C70, 33D80
\par}

{\small\tableofcontents\par}

\setcounter{footnote}{0}
\renewcommand{\thefootnote}{\arabic{footnote}}

\section{Introduction}
\label{sec 1}
A KZ-type flat connection is a basic structure in representation theory, mathematical physics, and the enumerative geometry of Nakajima varieties. When such a connection is reduced to a field of characteristic $p$, it acquires a new family of commuting endomorphisms: the $p$-curvature operators. The purpose of this paper is to diagonalize these operators by developing a characteristic $p$ analogue of the Bethe ansatz method, based on an integral representation of the connection.

\smallskip

The general classical setting over $\C$ is as follows. Let $V$ be a complex vector space, $z=(z_1,\dots,z_n)$
base coordinates, and let $\ka\in\C^\times$. Define a connection $\nabla$ with fiber $V$ by the differential operators
\bea
\nabla_{m} = \frac{\der}{\der{z_m} }- \frac 1\ka H_m\,,\qquad m=1,\dots,n,
\eea
where $(H_m)$ are $(\End V)$-valued functions in $z$.
Assume that this connection is flat for every nonzero $\ka$. A flat section $I(z)$ is a $V$-valued solution of the system of equations
\beq
\label{qqKZ}
\frac{\der I}{\der{z_m}} =\frac 1\ka H_m \,I ,
\qquad m=1,\dots,n.
\eeq
A fundamental problem is to find integral representations for the flat sections, or equivalently, to realize the connection $\nabla$ as a suitable Gauss--Manin connection. The flatness of $\nabla$ implies, in particular, that
\bea
[H_l,H_m]=0
\eea
for all $l$ and $m$, and one naturally seeks the joint eigenvectors and eigenvalues of the commuting operators $(H_m)$.

\smallskip

Equations \eqref{qqKZ} arise in several guises: as KZ equations in conformal field theory and representation theory, and as quantum differential equations in the equivariant cohomology of Nakajima varieties. The commuting operators $(H_m)$ specialize to Gaudin Hamiltonians in the theory of quantum integrable models and to operators of quantum multiplication in that same cohomology theory; see, for instance, \cite{KZ, SV1, EFK, MO, AO}.

\smallskip

Integral representations for flat sections are known both for KZ equations and for the quantum differential equations of Nakajima varieties. Such a representation consists of a scalar (master) function $\Phi(t,z)$ and a $V$-valued (weight) function $w(t,z)$, where $t=(t_1,\dots,t_r)$ is a collection of auxiliary variables. The corresponding multidimensional hypergeometric integrals
\[
I(z) = \int \Phi(t,z)^{1/\ka}\,w(t,z)\,dt
\]
solve the system \eqref{qqKZ}.

The operators $(H_m)$ are classically diagonalized by the Bethe ansatz method as follows. One introduces the logarithmic derivatives
\bean
\label{pp}
\phi_l
=
\frac{\der \Phi}{\der t_l}\Big/\Phi \,,
\qquad
\psi_m
=
\frac{\der \Phi}{\der z_m}\Big/\Phi\,.
\eean
If a point $(t^0, z^0)$ solves the system of Bethe ansatz equations
\bean
\label{1ba}
\phi_l(t,z)=0,
\qquad l=1,\dots,r,
\eean
then the evaluated weight function $w(t^0,z^0)$, if nonzero, is an eigenvector of $H_m(z^0)$ with eigenvalue $\psi_m(t^0,z^0)$:
\bean
\label{HW}
H_m(z^0)\,w(t^0,z^0)=\psi_m(t^0,z^0)\,w(t^0,z^0),
\qquad m=1,\dots,n.
\eean

\smallskip

When moving to a field $\K$ of characteristic $p>0$, the study of linear differential equations and flat connections is governed by the $p$-curvature operators. The $p$-curvature operators, defined by $C_m = (\nabla_m)^p$, measure the obstruction to the existence of horizontal sections. By a fundamental theorem of Cartier, the connection has a full set of flat sections if and only if its $p$-curvature vanishes identically. When the $p$-curvature is nonzero, the operators $(C_m)$ define a purely characteristic $p$ geometric invariant. Since the operators $(C_m)$ commute with the connection and with each other, understanding their simultaneous eigenspaces is a central problem in the geometry of characteristic $p$ connections.

\vsk.2>

In characteristic $p$, the classical analytic picture breaks down: transcendental master functions do not exist, and standard topological integration is unavailable. The primary goal of this paper is to establish an algebraic characteristic $p$ analogue of the multidimensional hypergeometric integral, and to demonstrate that this construction provides exact simultaneous eigenvectors and explicit eigenvalues for the $p$-curvature operators of non-abelian connections.

The core algebraic mechanism replaces classical integration over the auxiliary variables $t_1, \dots, t_r$ with the iterated differential operator $\Delta = \prod_{j=1}^r (\nabla^o_{t_j})^{p-1}$, where the auxiliary rank-one connection $\nabla^o$ is defined by
\bea
\nabla^o_{t_l}
&=&
\der_{t_l} + \frac 1\ka\phi_l(t,z), \qquad l=1,\dots, r,
\\
\nabla^o_{z_m}
&=&
\der_{z_m} + \frac 1\ka \psi_m(t,z), \qquad m=1,\dots, n,
\eea
where $\der_{t_l} := \frac{\der}{\der t_l}$, \, $\der_{z_m} := \frac{\der}{\der z_m}$.
This auxiliary connection encodes the logarithmic derivatives of the classical master function. Since $\der_x^{p-1}$ annihilates $x^k$ for $0\le k<p-1$ and maps $x^{p-1}$ to $(p-1)!=-1$, the operator $\Delta$ acts as an algebraic residue, closely mimicking the behavior of a contour integral.

Denote
\bean
\label{ID}
I = \Delta(w).
\eean
This is a $V$-valued function of $(t,z)$.

In characteristic $p$, the $p$-th powers $(\nabla^o_{t_l})^p$ and $(\nabla^o_{z_m})^p$ act as multiplication by scalar functions, which we denote by $\frac 1{\ka^p}\tilde \phi_l$ and $\frac 1{\ka^p}\tilde \psi_m$, respectively.

In this paper, we show that if $(t^0,z^0)$ satisfies the characteristic $p$ Bethe ansatz system
\bean
\label{gbae}
\tilde \phi_l(t,z)=0,
\qquad l=1,\dots,r,
\eean
then the vector $I(t^0,z^0)$, if nonzero,  is an eigenvector of the $p$-curvature operator $C_m(z^0)$ with eigenvalue $-\frac1{\ka^p}\tilde \psi_m(t^0,z^0)$:
\[
C_m(z^0)\,I(t^0,z^0)=-\frac 1{\ka^p}\tilde \psi_m(t^0,z^0)\,I(t^0,z^0),
\qquad m=1,\dots,n,
\]
see Corollary \ref{cor eig}, where $\bar\psi_m=\ka^{-p}\tilde\psi_m$, and Corollary \ref{cor ig} for the KZ connection.

In fact, the geometric statement is much richer. Given a solution $(t^0,z^0)$ of \eqref{gbae}, consider the ideal
\bea
\mathcal{J}_p = \langle (t_1-t_1^0)^p, \dots, (t_r-t_r^0)^p, (z_1-z_1^0)^p, \dots, (z_n-z_n^0)^p \rangle \ \subset \ \K[t,z]
\eea
and the resulting quotient algebra
\bea
A_{(t^0,z^0)} = \K[t,z]\big/ \mc J_p, \qquad \dim_\K A_{(t^0,z^0)} = p^{r+n}.
\eea
The Artinian local algebra $A_{(t^0,z^0)}$ is called
the Frobenius neighborhood of $(t^0, z^0)$. The connection $\nabla$ induces a connection on $V\ox A_{(t^0,z^0)}$, and the connection $\nabla^o$ induces a connection on $A_{(t^0,z^0)}$. We may therefore regard $I$ as an element of $V\ox A_{(t^0,z^0)}$, that is, as a $V$-valued function on the Frobenius neighborhood of $(t^0,z^0)$.

\vsk.2>
We show that if $(t^0,z^0)$ solves \eqref{gbae}, then the element $I$ is a local flat section of the extended connection over the Frobenius neighborhood:
\bea
\left(\der_{z_m}+ \frac 1\ka\psi_m - \frac 1\ka H_m\right) I
&=&
0,\qquad m=1,\dots, n,
\\
\left(\der_{t_l}+ \frac 1\ka\phi_l\right) I
&=&
0,\qquad l=1,\dots, r.
\eea
Furthermore, the element $I$ is an eigensection of the $p$-curvature operators $(C_m)$ with respective constant eigenvalues $(-\frac1{\ka^p}\tilde \psi_m(t^0,z^0))$, meaning
\bea
C_m\,I \,=\, -\, \frac 1{\ka^p}\tilde \psi_m(t^0,z^0)\, I, \qquad m=1,\dots,n.
\eea
See Theorem \ref{thm eig}.

\vsk.2>

Formula \eqref{ID} reveals a fundamental departure from the classical theory: the eigenvector $I(t^0,z^0)$ is not simply the evaluation $w(t^0,z^0)$ of the weight function at a Bethe ansatz solution, as it was in \eqref{HW}. Rather, the eigenvector is the evaluated characteristic $p$ integral $\Delta(w)$.

\vsk.2>

Remarkably, the scalar $p$-curvature functions $\tilde \phi_l(t,z)$ and $\tilde \psi_m(t,z)$ satisfy a special symmetry, relating them directly to the classical logarithmic derivatives evaluated at $p$-th powers:
\bean
\label{spsy}
\tilde \phi_l(t,z) = \phi_l(t_1^p,\dots,t_r^p,z_1^p,\dots,z_n^p),
\qquad
\tilde \psi_m(t,z) =  \psi_m(t_1^p,\dots,t_r^p,z_1^p,\dots,z_n^p),
\eean
where some parameters of $\phi_l$ and $\psi_m$ in \eqref{spsy} are modified by using the Artin--Schreier polynomial
$y^p-\ka^{p-1}y$.
See formulas \eqref{kzoo} and \eqref{tikz} for a precise comparison. See also the Motivating Example in Section \ref{n=3 p=3}.

\bigskip

Our construction naturally raises the usual Bethe-ansatz-type questions in this new algebraic setting:
\begin{itemize}
\item Is the constructed eigenvector nonzero?
\item Do all eigenvectors of the $p$-curvature operators arise in this way?
\item Can one compute their norms?
\item Are the resulting eigenvectors orthogonal?
\end{itemize}

%% begin edited

\subsection{Motivating Example}
\label{n=3 p=3}\rm

 Let \,$M_{\La_1}\ox M_{\La_2}\ox M_{\La_3}$
\,be the tensor product of three Verma modules
over the Lie algebra \,$\sll_2$ with highest weights $\La_1,\,\La_2$, and $\La_3$.
Let \,$V\subset M_{\La_1}\ox M_{\La_2}\ox M_{\La_3}$
\,be the three-dimensional subspace  with
 the basis
\be
f^{(1,0,0)}=f\<\>v_{\La_1}\!\ox v_{\La_2}\!\ox v_{\La_3}\>,\quad
f^{(0,1,0)}=v_{\La_1}\!\ox f\<\>v_{\La_2}\!\ox v_{\La_3}\>,\quad
f^{(0,0,1)}=v_{\La_1}\!\ox v_{\La_2}\!\ox f\<\>v_{\La_3}\>.
\ee
In this basis,
\begin{align*}
\Om^{(1,2)}=\,\Om^{(2,1)}=\begin{pmatrix}
(\La_1-2)\>\La_2/2 & \La_2 & 0\\[4pt]
\La_1 & \La_1\>(\La_2-2)/2 & 0\\[4pt]
0 & 0 & \La_1\<\>\La_2/2\end{pmatrix},
\\[8pt]
\Om^{(1,3)}=\,\Om^{(3,1)}=\begin{pmatrix}
(\La_1-2)\>\La_3/2 & 0 & \La_3 \\[4pt]
0 & \La_1\<\>\La_3/2 & 0\\[4pt]
\La_1 & 0 & \La_1\>(\La_3-2)/2\end{pmatrix},
\\[8pt]
\Om^{(2,3)}=\,\Om^{(3,2)}=\begin{pmatrix}
\La_2\<\>\La_3/2 & 0 & 0 \\[4pt]
0 & (\La_2-2)\>\La_3/2 & \La_3\\[4pt]
0 & \La_2 & \La_2\>(\La_3-2)/2\end{pmatrix},
\end{align*}
and the KZ connection is
\be
\nabla_m\>=\,\der_{z_m}\!-\>
\frac 1\ka\,\sum_{j\ne m}\>\frac{\Om^{(m,j)}}{z_m-z_j}\;,\qquad m=1,2,3\,.
\ee
The KZ connection preserves the two-dimensional subspace of singular vectors \;$\Sing\;V$
with a basis
\be
\La_2\<\>f^{(1,0,0)}\<-\La_1\<\>f^{(0,1,0)},\quad
\La_3\<\>f^{(1,0,0)}\<-\La_1\<\>f^{(0,0,1)}
\ee
for generic $\La_1,\La_2,\La_3$.

\smallskip
\noindent
In this example, \,$t$ is a single variable, and
\begin{align*}
\phi(t;z_1,z_2,z_3)\, &{}=\,\sum_{m=1}^3\,\frac{-\<\>\La_m}{t-z_m}\;,
\\[4pt]
\psi_m(t;z_1,z_2,z_3)\, &{}=\,\frac{-\<\>\La_m}{z_m-t}\>+\>
\sum_{j\ne m}\>\frac{\La_m\<\>\La_j/2}{z_m-z_j}\;,\qquad m=1,2,3\,.
\end{align*}
The $V$-valued weight function is
\be
w(t;z_1,z_2,z_3)\,=\,\frac{f^{(1,0,0)}}{t-z_1}\>+\>
\frac{f^{(0,1,0)}}{t-z_2}\>+\>\frac{f^{(0,0,1)}}{t-z_3}\;.
\ee
Let \,$p=3$\,. Then the \,$p$-curvature operators are
\,$C_m=\>(\nabla_m)^3$\>, \,$m=1,2,3$\,, and
\begin{align*}
\\[-16pt]
\tilde\phi(t;z_1,z_2,z_3)\, &{}=\,
\sum_{m=1}^3\,\frac{-\<\>h_3(\La_m)}{t^3\<-z_m^3}\,,
\\[4pt]
\tilde\psi_m(t;z_1,z_2,z_3)\, &{}=\,
\frac{-\<\>h_3(\La_m)}{z_m^3-t^3}\>+\>
\sum_{j\ne m}\>\frac{h_3(\La_m\<\>\La_j/2)}{z_m^3-z_j^3}\,,\qquad m=1,2,3\,,
\end{align*}
where \,$h_3(y)=y^3\<-\ka^2 y\,$ is the Artin--Schreier polynomial for $p=3$.

\smallskip

Corollaries \ref{cor sing} and \ref{cor ig} in this case read as follows.
Let \,$(t^0; z_1^0,z_2^0,z_3^0)$ \,be a solution to the Bethe ansatz equation
\be
\frac{h_3(\La_1)}{t^3\<-z_1^3}\,+\,\frac{h_3(\La_2)}{t^3\<-z_2^3}\,+\,
\frac{h_3(\La_3)}{t^3\<-z_3^3}\;=\,0\,.
\ee
Then the vector
\be
I(t^0;z_1^0,z_2^0,z_3^0)\,=\,
\biggl(\!\Bigl(\der_t+\>\frac 1\ka\>\phi(t;z_1,z_2,z_3)\Bigr)^{\!2}\<
w(t;z_1,z_2,z_3)\biggr)\bigg|_{(t;z_1,z_2,z_3)=(t^0; z_1^0,z_2^0,z_3^0)}
\ee
belongs to the subspace of singular vectors \;$\Sing\;V$ and, if nonzero, 
is an eigenvector of the \,$p$-curvature operators
\,$C_m(z_1^0,z_2^0,z_3^0)$ \,with respective eigenvalues
\be
-\>\frac 1{\ka^3}\left(
\frac{-\<\>h_3(\La_m)}{(z_m^0)^3-(t^0)^3}\>+\>
\sum_{j\ne m}\>\frac{h_3(\La_m\<\>\La_j/2)}{(z_m^0)^3-(z_j^0)^3}\right),\qquad m=1,2,3\,.
\ee

\smallskip
%% end edited

The paper is organized as follows.
In Section \ref{sec 2}, we introduce the notion of an integral representation for a flat connection in characteristic $p$. We develop a construction of eigenvectors and eigenvalues of the $p$-curvature operators of a flat connection admitting an integral representation.

In Sections \ref{sec 3}--\ref{sec 5}, we apply the general construction of Section \ref{sec 2} to
particular KZ-type differential connections. To keep the notation minimal and concrete, we focus on two primary examples. The first is the rational KZ connection associated with a tensor product of Verma modules over $\mathfrak{sl}_2$ in characteristic $p$, see Section \ref{sec 3}. The second is the compatible system of KZ and dynamical
connections
associated with a tensor product
of Verma modules over $\mathfrak{sl}_2$ in characteristic $p$, see Section \ref{sec 5}.

As an illustration, in Section \ref{sec 4} we consider a further special case of the KZ connection of Section \ref{sec 3}, where the $\mathfrak{sl}_2$ weights, as well as $\ka$, are restricted to the prime field $\F_p$. We show that in this case the system of Bethe equations \eqref{gbae} is trivial (every point $(t^0,z^0)$ satisfies the system), and our integration method globally produces polynomial flat sections of the KZ connection. In this specialized setting, our construction shares a similar flavor with the construction of polynomial solutions to the KZ and qKZ equations in characteristic $p$ presented in \cite{SV2, MV1, EV1, EV2}.

\smallskip

This paper is related to \cite{TV1, TV2}.
In \cite{TV1}, a qKZ-type discrete flat connection with multiplicative step $q$ is considered. The connection has monodromy operators when $q$ is a root of unity. The monodromy operators are commuting endomorphisms of the
discrete
connection. The eigensections and eigenvalues of the monodromy operators are constructed. The construction is based on the presentation of the discrete flat connection as a discrete Gauss--Manin connection.

In \cite{TV2}, a qKZ-type additive discrete flat connection in characteristic $p$ is considered. The connection has $p$-curvature operators, which are commuting endomorphisms of the discrete
connection. A Bethe-ansatz-type construction of eigensections and eigenvalues of the $p$-curvature operators is presented. An eigensection of the $p$-curvature operators is a discrete hypergeometric sum over the finite lattice $\Z^r/p\Z^r$. Thus the $p$-curvature eigensections are constructed by a finite, discrete analogue of a hypergeometric integral.

\smallskip

\noindent\textbf{Acknowledgments.}
The second author thanks IH\'ES for its hospitality during May--June 2026, when this paper was developed.

The authors thank P.\,Etingof for useful discussions.

\section{Eigenvectors of \texorpdfstring{$p$}{p}-curvature operators}\label{sec 2}

\subsection{Flat connection}
\label{sec 2.1}

Let $p$ be an odd prime, $n$ a positive integer.
Let $\K$ be a field of characteristic $p$ and $V$ a finite-dimensional vector space over $\K$.
Let $\A^n$ be the $n$-dimensional affine space over $\K$ with coordinates $z=(z_1,\dots,z_n)$.
Let $H_m(z)$ for $m=1,\dots,n$ be $\End V$-valued rational functions
in $z$. Consider the connection $\nabla$ on the trivial bundle over $\A^n$ with fiber $V$ defined by
\bea
\nabla_m = \der_{z_m} + H_m(z), \qquad m=1,\dots,n \,,
\eea
where $\der_{z_m} = \frac{\der}{\der z_m}$. The connection has singularities at the poles of the functions
$(H_m)$.
Assume that $\nabla$ is flat, that is, $[\nabla_m,\nabla_l]=0$ for all $l,m$.
\vsk.2>

The $p$-curvature operators of the connection $\nabla$ are defined by
\bea
C_m = \left(\nabla_m\right)^p, \qquad m=1,\dots,n.
\eea
Each $C_m$ defines an endomorphism of the connection, that is, $[C_m,\nabla_l] = 0$ for
$l=1, \dots, n$. The action of $C_m$ on sections of the bundle commutes with multiplication of sections by functions on
the base. The operator $C_m$ is an $\End V$-valued rational function in $z$.

If $s $ is a flat section of $\nabla$, then $s$ is annihilated by the $p$-curvature operators,
that is, $C_m s=0$ for all $m$.

\begin{exmp}
Let $n=1$, $V=\K$, $\nabla =\der_x + f(x)$ where $f(x)\in \K(x)$. Then
\bean
\label{ex1}
(\der_x + f)^p = f^p+f^{(p-1)}.
\eean
\end{exmp}

\begin{exmp}
\label{ex2}

Let $n=1$, $f(x)\in\K(x)$, $H(x) \in (\End V)(x)$.
Consider the three connections:
\bea
\nabla_f := \der_x + f(x),
\qquad
\nabla_H := \der_x + H(x),
\qquad
\nabla_{f\on{Id}_V+H} := \der_x + f(x)\on{Id}_V +H\,.
\eea
Then the corresponding $p$-curvatures satisfy the relation:
\bean
\label{C-rel}
C_{f\on{Id}_V+H} =C_f \on{Id}_V + C_H\,.
\eean
Indeed, the third connection $\nabla_{f\on{Id}_V+H}$ is identified with the tensor product of the first two connections:
$\nabla_f \ox \on{Id}_V + \on{Id}_\K\ox \nabla_H$. Then the $p$-curvature of the tensor-product connection
equals $C_f \ox \on{Id}_V + \on{Id}_\K\ox C_H$, which gives \eqref{C-rel}.

\end{exmp}

Proofs of these and other properties of $p$-curvature operators can be found in \cite{K1, K2};
also see, for example, \cite{EV2}.

\subsection{Auxiliary flat connection of rank one}
\label{sec 2.2}

Let $r$ be a positive integer. Let $\A^{r+n}$ be the affine space over $\K$ of dimension $r+n$ with coordinates
$(t, z)$, where $t=(t_1,\dots, t_r)$.
Let $\phi_l(t,z)$ for $l=1,\dots,r$ and $\psi_m(t,z)$ for $m=1,\dots,n$
be elements of $\K(t,z)$, that is, scalar rational functions in $(t,z)$. Define a rank-one connection
$\nabla^o$ on $\A^{r+n}$ by the formulas:
\bea
\nabla^o_{t_l}
&=&
\der_{t_l} + \phi_l(t,z), \qquad l=1,\dots, r,
\\
\nabla^o_{z_m}
&=&
\der_{z_m} + \psi_m(t,z), \qquad m=1,\dots, n.
\eea
The connection $\nabla^o$ has singularities at the poles of the functions $(\phi_l,\psi_m)$.
Assume that $\nabla^o$ is flat.

\begin{lem}
\label{lem: tilde}
The $p$-curvature operators of $\nabla^o$ are given by the formulas:
\bean
\label{p l}
(\nabla_{t_l}^o)^p
&=&
\bar \phi_l,\qquad l=1,\dots,r,
\\
\label{p m}
(\nabla_{z_m}^o)^p
&=&
\bar \psi_m, \qquad m=1,\dots,n,
\eean
where
\bean
\label{f l}
\bar \phi_l
&=&
\phi_l^p + \der_{t_l}^{p-1}(\phi_l),
\\
\label{f m}
\bar \psi_m
&=&
\psi_m^p + \der_{z_m}^{p-1}(\psi_m).
\eean
Furthermore, each of the functions $(\bar\phi_l, \bar\psi_m)$ is a function of
$(t^p,z^p) := (t_1^p,\dots, t_r^p, z_1^p, \dots, z_n^p)$.

\end{lem}

\begin{proof}
Formulas \eqref{f l} and \eqref{f m} follow from \eqref{ex1}.

For any $l$ and $ m $, on the one hand, we have
\bea
[\nabla^o_{t_l}, (\nabla^o_{z_m})^p]=0,
\eea
and, on the other hand,
\bea
[\nabla^o_{t_l}, \bar \psi_m] = \der_{t_l}\bar \psi_m\,.
\eea
Hence $\der_{t_l}\bar \psi_m=0$, and $\bar \psi_m$, as a function of $t_l$, is a function of $t_l^p$.
The remainder of the lemma is proved similarly.
\end{proof}

The Bethe ansatz equations are the system of equations
\bean
\label{bae}
\bar \phi_i(t,z) = 0, \qquad i=1,\dots,r.
\eean

\vsk.2>

Given $(t^0,z^0)$, consider the ideal
\bean
\label{Jp}
\mathcal{J}_p = \langle (t_1-t_1^0)^p, \dots, (t_r-t_r^0)^p, (z_1-z_1^0)^p, \dots, (z_n-z_n^0)^p \rangle \ \subset \ \K[t,z]
\eean
and the quotient
\bea
A_{(t^0,z^0)} = \K[t,z]\big/ \mc J_p, \qquad \dim_\K A_{(t^0,z^0)} = p^{r+n}.
\eea

A rational function $f\in\K(t,z)$ is called {\it regular} at $(t^0,z^0)$ if
$f=a/b$ with $a,b\in\K[t,z]$ and $b(t^0,z^0)\ne0$. The maximal ideal of the local
algebra $A_{(t^0,z^0)}$ is generated by the nilpotent elements $t_l-t_l^0$,
$z_m-z_m^0$, hence $b$ is invertible in $A_{(t^0,z^0)}$ whenever $b(t^0,z^0)\ne0$.
Thus every function regular at $(t^0,z^0)$ has a well-defined image in
$A_{(t^0,z^0)}$, and this assignment commutes with the operators $\der_{t_l}$,
$\der_{z_m}$. A $V$-valued or $\End V$-valued rational function is called regular
at $(t^0,z^0)$ if all its coordinates with respect to a basis are regular.

We call $(t^0,z^0)$ a {\it regular point} if the functions $H_m$, $\phi_l$,
$\psi_m$ are regular at $(t^0,z^0)$. At a regular point the connections $\nabla$
and $\nabla^o$ induce connections on $V\ox A_{(t^0,z^0)}$ and $A_{(t^0,z^0)}$,
respectively.

We retain the notation $\nabla$ and $\nabla^o$ for the induced connections. The induced connections remain flat.

\vsk.2>
Denote
\bean
\label{delta}
\Delta = \prod_{j=1}^r \left(\nabla^o_{t_j}\right)^{p-1} .
\eean
We have the following useful observation.

\begin{lem}
\label{lem Sto}
Let $(t^0,z^0)$ be a regular point which is a solution to the system \eqref{bae}.
Let $1\leq i\leq r$ and $g\in A_{(t^0,z^0)}$. Then
$\Delta \,\nabla^o_{t_i} (g)$
is the zero element of $A_{(t^0,z^0)}$.
\end{lem}

\begin{proof}
We have
\bea
\Delta \, \nabla^o_{t_i} (g)
=\left(\nabla^o_{t_i}\right)^p
\left(\prod_{j\ne i}^r \left(\nabla^o_{t_j}\right)^{p-1}\right) (g)
= \bar \phi_i \cdot
\left(\prod_{j\ne i}^r \left(\nabla^o_{t_j}\right)^{p-1}\right) (g).
\eea
By Lemma \ref{lem: tilde}, $\bar \phi_i$ is a function of $(t^p,z^p)$.
   Since $\phi_i$ is regular at $(t^0,z^0)$, so is $\bar\phi_i$; write
   $\bar\phi_i=a/b$ with $b(t^0,z^0)\ne0$. Then $\bar\phi_i=ab^{p-1}/b^p$, where
   $b^p\in\K[t^p,z^p]$ and $b^p(t^0,z^0)=b(t^0,z^0)^p\ne0$. The numerator
   $ab^{p-1}=\bar\phi_i\,b^p$ is a polynomial annihilated by all $\der_{t_l}$,
   $\der_{z_m}$, hence it also lies in $\K[t^p,z^p]$.

Because we are operating in the quotient $A_{(t^0,z^0)}$ defined by $\mc J_p$, the $p$-th powers of the variables are constant: $t_l^p \equiv (t_l^0)^p$ and $z_m^p \equiv (z_m^0)^p$. Thus, $\bar \phi_i$ reduces to the constant $\bar \phi_i(t^0,z^0)$, which is zero by the assumption that $(t^0,z^0)$ solves the Bethe ansatz equations. Therefore, $\bar \phi_i$ induces the zero element in $A_{(t^0,z^0)}$.
\end{proof}

\subsection{Integral representation for \texorpdfstring{$\nabla$}{the connection}}
\label{sec 2.3}

Let the $\End V$-valued functions $(H_m)$ define a flat connection
$\nabla$ as in Section \ref{sec 2.1}. Let the scalar functions
$(\phi_l, \psi_m)$ define a rank-one flat connection $\nabla^o$ as in Section \ref{sec 2.2}.

Define a new connection, called the {\it extended connection}, with base $\A^{r+n}$ and fiber $V$ by the following differential operators:
\bean
\label{ncon}
&&
\nabla^o_{t_l}, \qquad l=1,\dots,r,
\\
&&
\nabla_{z_m}^o+H_m, \qquad m=1,\dots,n.
\eean
This connection is flat.

\vsk.2>

Let $w$ and $g_{l,m}$, for $l=1, \dots, r$ and $m=1,\dots,n$, be $V$-valued rational functions in
$(t,z)$ such that
\bean
\label{IR}
\left(\nabla^o_{z_m} + H_m\right)\!(w) = \sum_{l=1}^r \nabla^o_{t_l}\left(g_{l,m}\right), \qquad m=1,\dots,n.
\eean
Such a
collection of functions is called an {\it integral representation} for the flat connection $\nabla$.
The function $w$ is called the {\it weight function.}

\vsk.2>

Given a regular point $(t^0,z^0)$ at which $w$ and all $g_{l,m}$ are regular,
define an element $I\in V\ox A_{(t^0,z^0)}$ by the formula
\bean
\label{elt I}
I = \Delta(w).
\eean

\begin{thm}
\label{thm eig}

Let $(t^0,z^0)$ be a regular point at which $w$ and all $g_{l,m}$ are regular, and which is a solution
 to the system of Bethe ansatz equations:
\bean
\label{BEA}
\bar \phi_l(t,z) = 0,
\qquad l=1,\dots,r.
\eean
Then the element $I \in V\ox A_{(t^0,z^0)}$ is a flat section of the extended connection. That is,
\bean
\label{fecl}
\left(\der_{z_m}+ \psi_m + H_m\right) I
&=&
0,\qquad m=1,\dots, n,
\\
\label{fecm}
\left(\der_{t_l}+ \phi_l\right) I
&=&
0,\qquad l=1,\dots, r.
\eean
Furthermore, the element $I $ is an eigensection of the $p$-curvature operators $(C_m)$
with respective eigenvalues $(-\bar \psi_m(t^0,z^0))$, that is,
\beq
\label{eigv}
C_m\,I \,=\, -\, \bar \psi_m(t^0,z^0)\, I, \qquad m=1,\dots,n.
\eeq

\end{thm}

\begin{proof}

Applying $\Delta$ to both sides of equation \eqref{IR}, we obtain
\bean
\label{2.13}
\left(\der_{z_m}+ \psi_m + H_m\right)I
=
\sum_{l=1}^r \,\bar \phi_l\cdot \left(\prod_{j\ne l}^r \left(\nabla^o_{t_j}\right)^{p-1}\right) \!
\left(g_{l,m}\right),
\eean
where $ \bar \phi_l $ acts as the zero element in $A_{(t^0,z^0)}$. Formula \eqref{fecl} is proved.

We have
\bea
\nabla^o_{t_l} I = \left(\nabla^o_{t_l}\right)^p \left(\prod_{j\ne l} \left(\nabla^o_{t_j}\right)^{p-1}\right) (w)
=\bar \phi_l \cdot \left(\prod_{j\ne l} \left(\nabla^o_{t_j}\right)^{p-1}\right) \!(w).
\eea
This element is zero in $V \ox A_{(t^0,z^0)}$ since $\bar \phi_l$ is zero in $A_{(t^0,z^0)}$.
This proves formula \eqref{fecm}.

Applying $\left(\der_{z_m}+ \psi_m + H_m\right) ^{p-1}$ to formula \eqref{fecl} yields
\bea
0=\left(\der_{z_m}+ \psi_m + H_m\right)^p I =
\left(C_m + \bar \psi_m\right )I,
\eea
where $\bar \psi_m$ equals the constant $\bar \psi_m(t^0,z^0)$ in $A_{(t^0,z^0)}$ since
$\bar \psi_m$ is a function of $(t^p,z^p)$ by Lemma \ref{lem: tilde},
see the argument in the proof of Lemma \ref{lem Sto}.
This proves formula \eqref{eigv}. Theorem \ref{thm eig} is proved.
\end{proof}

\begin{cor}
\label{cor eig}

Let $(t^0,z^0)$ be a regular point at which $w$ and all $g_{l,m}$ are regular, and which is a solution
to the system of Bethe ansatz equations \eqref{BEA}.
Then the vector $I(t^0,z^0) \in V$, if nonzero, is an eigenvector of the $p$-curvature operators $(C_m(z^0))$
with respective eigenvalues $(-\bar \psi_m(t^0,z^0))$, that is,
\bea
C_m(z^0)\,I(t^0,z^0) \,=\, -\, \bar \psi_m(t^0,z^0)\, I(t^0,z^0), \qquad m=1,\dots,n.
\eea

\end{cor}

Theorem \ref{thm eig} describes the properties of the element $I\in V\ox A_{(t^0,z^0)}$ if $(t^0,z^0)$ is a solution
to the system of Bethe ansatz equations. Notice that the elements
$I, \phi_l, \psi_m$ in Theorem \ref{thm eig} depend on $t$ and $z$ in the quotient, while $H_m, C_m$ depend on $z$ only.
In particular, we have $C_m(z) I(t,z) = -\bar \psi_m(t^0,z^0) I(t,z)$.
We show below that the geometric direction of the eigenvector $I(t,z)$ in $V$ does not depend on $t$.

\vsk.2>
More precisely,
let $e_1,\dots,e_d$ be a basis of $V$, and write $I = \sum_{j=1}^d I_j e_j$ where $I_j \in A_{(t^0,z^0)}$. Assume that
$I_{j_0}(t^0,z^0) \ne 0$ for some $j_0$. For any $j$, denote the ratio $I_j/I_{j_0}\in A_{(t^0,z^0)}$ by $R_j$. We may express $R_j$ as a polynomial in the quotient:
\bea
R_j = \sum_{0\leq i_1,\dots, i_r, k_1,\dots,k_n\leq p-1}
c_{i_1,\dots, i_r, k_1,\dots,k_n}
\prod_{l=1}^r (t_l-t_l^0)^{i_l} \prod_{m=1}^n(z_m-z_m^0)^{k_m},
\eea
where $c_{i_1,\dots, i_r, k_1,\dots,k_n} \in \K$.

\begin{lem}
\label{lem dit}

If $i_1+\dots+ i_r>0$, then $c_{i_1,\dots, i_r, k_1,\dots,k_n} = 0$. In other words, the ratio
$R_j$ does not depend on $t$,
and the direction of the element $I$ depends on $z$ only.

\end{lem}

\begin{proof}
By formula \eqref{fecm}, the standard derivative in the quotient satisfies $\der_{t_l}I = -\phi_l I$. Because $\phi_l$ is a scalar function, this implies $\der_{t_l}I_j = -\phi_l I_j$ for all components $j$.
Applying the quotient rule, we obtain
\bea
\der_{t_l}R_j = \frac{\left(\der_{t_l} I_j\right) I_{j_0} - I_j \left(\der_{t_l}I_{j_0}\right)}{I_{j_0}^2} = \frac{(-\phi_l I_j) I_{j_0} - I_j (-\phi_l I_{j_0})}{I_{j_0}^2} = 0.
\eea
Since $\der_{t_l} R_j \equiv 0$ in $A_{(t^0,z^0)}$ for all $l=1,\dots,r$, and the quotient ideal only restricts powers greater than or equal to $p$, all coefficients in $R_j$ involving strictly positive powers of $(t_l - t_l^0)$ must be zero.
\end{proof}

\subsection{Global sections}

A nonzero rational function $\Phi(t,z)\in\K(t,z)$ is called a master function of the connection $\nabla^o$ if
\bean
\label{ast}
\frac{\der_{t_l} \Phi}{\Phi} = \phi_l(t,z), \qquad \frac{\der_{z_m} \Phi}{\Phi} = \psi_m(t,z)
\eean
for all $l, m$. If $g\in\K(t,z)$, then
\bean
\label{nndc}
\Phi \nabla^o_{t_l}(g) = \der_{t_l}(\Phi g),
\qquad
\Phi \nabla^o_{z_m}(g) =\der_{z_m}(\Phi g),
\eean
for all $l,\,m$. Therefore, $\Phi^{-1}$ is a nonzero flat section of $\nabla^o$,
and the connection $\nabla^o$ has zero $p$-curvature, that is, 
\bean
\label{p-c=0}
\bar \phi_l = 0, \qquad 
\bar \psi_m =0
\eean
for all $l,m$.

Conversely, assume that \eqref{p-c=0} holds. Let $(t^0,z^0)$ be a point at which
$\phi_l$, $\psi_m$ are regular, and let $g\in\K(t,z)$ be regular at $(t^0,z^0)$
with $g(t^0,z^0)\ne0$. Then the function $\Psi$ defined by 
\bean
\label{Psi}
\phantom{aaa}
\Psi = 
\left(\prod_{i=1}^r \left(\nabla^o_{t_i}\right)^{p-1}\right)
\left(\prod_{j=1}^n \left(\nabla^o_{z_j}\right)^{p-1}\right)
\left(\prod_{i=1}^r (t_i-t^0_i)^{p-1}\right)
\left(\prod_{j=1}^n(z_j-z_j^0)^{p-1}\right) g
\eean
is a flat section of $\nabla^o$, and $\Psi(t^0,z^0)=(-1)^{r+n}g(t^0,z^0)\ne0$.
Hence $\Psi^{-1}$ is a master function of $\nabla^o$.

\medskip

Return to the  connections $\nabla$, $\nabla^o$, and  the  integral representation 
of Section \ref{sec 2.3}.

\begin{thm}
\label{thm global}

 Assume that 
\eqref{p-c=0} holds, and $\Phi$ is a master function.
Define  $J \in V\ox \K(t,z)$ by the formula:
\bean
\label{elt J}
J = \Phi \,I =\Phi \,\Delta (w)  = \left(\prod_{i=1}^r \der_{t_i}^{p-1} \right)(\Phi \,w).
\eean 
 Then $J$ is a flat section of $\nabla$\,:
\bean
\label{fsP}
\left(\der_{z_m} + H_m\right) J=0, \qquad m=1,\dots, n,
\eean
and also
\bean
\label{t-fl}
\der_{t_l}J=0, \qquad l=1,\dots,r.
\eean

\end{thm}

\begin{proof} 

Equation \eqref{fecl} holds for any values of $(t,z)$ since \eqref{p-c=0} holds.
Multiplying both sides of \eqref{fecl} by $\Phi$ we obtain
\bean
0=\Phi \left(\der_{z_m}+ \psi_m + H_m\right) I = \left(\der_{z_m} +H_m\right) \Phi I.
\eean
This proves \eqref{fsP}. Equation \eqref{t-fl} follows similarly from \eqref{fecm}.
\end{proof}

\subsection{Eigensection \texorpdfstring{$I=\Delta(w)$}{I = Delta(w)} as a characteristic \texorpdfstring{$p$}{p} hypergeometric integral}

In the complex analytic setting, solutions to the KZ equations and related connections are classically constructed via multidimensional hypergeometric integrals of the form
\bea
\int_\gamma \Phi(t,z) w(t,z) dt_1 \dots dt_r,
\eea
where $\Phi(t,z)$ is a multi-valued master function and $\gamma$ is a suitable integration cycle ensuring that boundary terms vanish. The logarithmic derivatives of the master function define a rank-one flat connection:
\bea
\frac{\der_{t_l} \Phi}{\Phi} = \phi_l(t,z), \qquad \frac{\der_{z_m} \Phi}{\Phi} = \psi_m(t,z).
\eea

In a field of characteristic $p$, transcendental functions such as $\Phi(t,z)$ do not exist
in general, and standard integration over topological cycles is unavailable. However, the differential operator $\der_x^{p-1}$ naturally serves as an algebraic analog of integration. Specifically, the operator $\der_x^{p-1}$ acts on powers of $x$ by the formula:
\bea
\der_x^{p-1} (x^k) = 0 \quad \text{for all } 0 \le k < p-1, \qquad \text{and} \qquad
\der_x^{p-1} (x^{p-1}) = (p-1)! \equiv -1 \pmod p.
\eea
Because $\der_x^{p-1}$ annihilates all monomials $x^k$ with $k\not\equiv p-1 \pmod p$, evaluating $\der_x^{p-1}$ at a point 
provides an algebraic analogue of residue extraction and contour integration.

To make this formal, consider a generic one-dimensional rank-one connection $\nabla = \der_x + f(x)$ where $f(x) = \sum_{i=1}^N \frac{a_i}{x-b_i}$. Formally, one may introduce a multi-valued integrating factor
\bea
\Phi(x) = \prod_{i=1}^N (x-b_i)^{a_i}.
\eea
Because $f(x)$ is the formal logarithmic derivative of $\Phi(x)$, the connection operator can be rewritten as:
\bea
\nabla = \frac{1}{\Phi(x)} \circ \der _x \circ \Phi(x).
\eea
By iterating this operator $p-1$ times, we obtain:
\bea
\nabla^{p-1}(g(x)) = \frac{1}{\Phi(x)} \der_x^{p-1} \Big( \Phi(x) g(x) \Big).
\eea
This expression provides a purely algebraic characteristic $p$ realization of the complex hypergeometric integral $\int \Phi(x) g(x) dx$.

Returning to the multidimensional setting of Theorem \ref{thm eig}, the operator $\Delta = \prod_{j=1}^r (\nabla^o_{t_j})^{p-1}$ replaces integration over the auxiliary variables $t_1, \dots, t_r$. Furthermore, the Bethe ansatz equations \eqref{BEA}, which require the $p$-curvatures $\bar \phi_l$ to vanish, provide the exact characteristic $p$ geometric analog of the condition that the integration cycle $\gamma$ is closed (i.e., that the boundary evaluation terms of Stokes' theorem vanish).

Consequently, the element $I = \Delta(w)$ defined in the Frobenius neighborhood $A_{(t^0,z^0)}$ represents the systematic characteristic $p$ adaptation of a complex multidimensional hypergeometric integral, directly linking the algebraic geometry of $p$-curvatures to the classical Bethe ansatz. The applications of Theorem \ref{thm eig} in the next sections will illustrate this analogy more explicitly.

\subsection{Applications}

Given a flat connection in characteristic $p$,
one can construct eigenvectors and eigenvalues of the $p$-curvature operators
of this connection whenever an integral representation of the type described in Section \ref{sec 2.3} is available. Such integral representations are known
to exist for the systems of KZ--type differential equations and for the
quantum differential equations in the equivariant cohomology
of Nakajima varieties.

To keep the notation minimal, we focus on two examples. The first is the rational KZ connection
associated with
a tensor product of Verma modules over $\mathfrak{sl}_2$ in characteristic $p$.
The second example is the compatible system of rational KZ and dynamical connections associated with
a tensor product of Verma modules over $\mathfrak{sl}_2$ in characteristic $p$.

\section{KZ connection}
\label{sec 3}

\subsection{Verma modules}

Let $p$ be an odd prime and $\K$ a field of characteristic $p$.
Consider the Lie algebra \,$\frak{sl}_2$ \,over $\K$ with the standard generators
$e$, $f$, $h$ such that $[h,e]= 2e$, $[h,f]=-2f$, $[e,f] = h$.

Let $\La \in\K$. Let $M_\La$ denote the Verma module over \,$\frak{sl}_2$
\,with highest weight \,$\La$ \,and highest weight vector $v_\La$:
\be
e\<\>v_\La=0,\qquad h\<\>v_\La=\La\>v_\La\,.\ee
A basis of $M_\La$ is formed by the vectors $f^{\<\>k}v_\La$, $k\in\Z_{\geq 0}$.

Let $\La_1,\dots,\La_n\in\K$.
We have the degree decomposition
\bea
\tsize\bigotimes^n_{j=1} M_{{\La_j}} =
\bigoplus_{r=0}^\infty
\left(\bigotimes^n_{j=1} M_{{\La_j}}\right)\left[r\right]\,.
\eea
A basis of a weight subspace
\,$V=\left(\ox^n_{j=1} M_{{\La_j}}\right)\left[r\right]$
\,is formed by the vectors
\beq
\label{basis}
f^{\<\>\vec r}:= f^{\<\>r_1}v_{\La_1} \ox\dots\ox f^{\<\>r_n}v_{\La_n}\,,
\eeq
where $\vec r=(r_1,\dots,r_n)$, $r_1+\dots+r_n = r$.
The set of such indices $\vec r$ is denoted by $\mc I_r$.

\smallskip

The kernel of the map
\bea
\tsize
e : \left(\bigotimes^n_{j=1} M_{{\La_j}}\right)\left[r\right]
\to \left(\bigotimes^n_{j=1} M_{{\La_j}}\right)\left[r-1\right],
\
u\mapsto eu,
\eea
is denoted by
$\Sing \left(\bigotimes^n_{j=1} M_{{\La_j}}\right)\left[r\right]$ and called the subspace of singular vectors of degree $r$.

\vsk.2>
Fix a degree subspace
\bea
\tsize V=\left(\bigotimes^n_{j=1} M_{{\La_j}}\right)\left[r\right]\,.
\eea

\subsection{KZ connection}

Denote
\bean
\label{Casimir}
\Omega = e \otimes f + f \otimes e +
\frac{1}{2} h \otimes h\ \in\ \slt \ox\slt\,.
\eean
Let $z=(z_1,\dots,z_n)$ be variables. Define
the Gaudin Hamiltonians
by the formulas:
\bea
H_m(z_1,\dots,z_{n})\,=\, - \sum_{j\ne m}\frac{\Om^{(m,j)}}{z_m-z_j}\,,
\quad m=1,\dots,n.
\eea
These are elements of $(\End V)(z)$. Fix
\bea
\ka\,\in\,\K^\times.
\eea
Define the KZ connection $\nabla$ on the trivial bundle with base $\A^{n}$ and fiber $V$
by the formulas:
\bean
\label{kz}
\nabla_m\>=\,\der_{z_m} +\frac 1\ka H_m\,,\qquad m=1,\dots,n.
\eean
The connection is flat, see \cite{KZ, EFK}. The connection $\nabla$ has singularities at the poles of 
$(H_m)$.

\subsection{Auxiliary flat connection of rank one}

Let $\A^{r+n}$ be the affine space over $\K$ of dimension $r+n$ with coordinates
$(t, z)$.
Consider the rank-one connection $\nabla^o$ with base $\A^{r+n}$ defined by the
formulas:
\bea
\nabla^o_{t_l}
&=&
\der_{t_l} +\frac1\ka \phi_l, \qquad l=1,\dots, r,
\\
\nabla^o_{z_m}
&=&
\der_{z_m} + \frac1\ka \psi_m, \qquad m=1,\dots, n,
\eea
where
\bean
\label{kzoo}
\phi_l
&=&
\sum_{m=1}^n \frac{-\La_m}{t_l-z_m} + \sum_{j\ne l} \frac 2{t_l-t_j}\,, \quad l=1,\dots,r\,,
\\
\notag
\psi_m
&=&
\sum_{l=1}^r \frac {-\La_m}{z_m-t_l}+ \sum_{i\ne m} \frac{\La_m\La_i/2}{z_m-z_i} \,, \quad
m=1,\dots,n\,.
\eean
The connection $\nabla^o$ has singularities at the poles of the functions $(\phi_l,\psi_m)$.

\begin{lem}
The connection $\nabla^o$ is flat.
\end{lem}

\begin{proof}

The proof is by inspection.
\end{proof}

The polynomial
\bea
h_p(y) = y^p-\ka^{p-1}y
\eea
is called the Artin--Schreier polynomial with parameter $\ka$.

\begin{lem}
The $p$-curvature operators of $\nabla^o$ are given by the formulas:
\bea
(\nabla_{t_l}^o)^p
&=&
\frac 1{\ka^p}\,
\tilde \phi_l,\qquad l=1,\dots,r,
\\
(\nabla_{z_m}^o)^p
&=&
\frac 1{\ka^p}\,
\tilde \psi_m, \qquad m=1,\dots,n,
\eea
where
\bean
\label{tikz}
\tilde \phi_l
&=&
 \sum_{m=1}^n \frac{-h_p(\La_m)}{t_l^p-z_m^p} + \sum_{j\ne l} \frac {h_p(2)}{t_l^p-t_j^p}
\,, \quad l=1,\dots,r\,,
\\
\notag
\tilde \psi_m
&=&
 \sum_{l=1}^r \frac {-h_p(\La_m)}{z_m^p-t_l^p}+ \sum_{i\ne m} \frac{h_p(\La_m\La_i/2)}{z_m^p-z_i^p} 
 \,, \quad
m=1,\dots,n\,.
\eean

\end{lem}

\begin{proof}
The lemma follows from formula \eqref{ex1}.
\end{proof}

\subsection{Weight function}

The $V$-valued weight function
\bea
w(t,z)\,=\>\sum_{\vec r\in \mc I_r} w_{\vec r}(t;z)\,f^{\<\>\vec r}
\eea
is defined by the formula
\bean
\label{Wkz}
w_{\vec r}(t,z) = \frac{1}{r_1! \dots r_n!}
\Sym_{t_1,\dots,t_r} \left[ \prod_{s=1}^{n}
\prod_{i=1}^{r_s}
\frac{1}{t_{ r_1 + \dots + r_{s-1}+i} - z_s}
\right] ,
\eean
where for a function $f(t;z)$, we define
\be
\on{Sym}_{t_1,\dots,t_r}f(t;z) =
\sum_{\si\in S_r} f(t_{\si(1)}, \dots,t_{\si(r)}; z)\,.
\ee
The expression in brackets in \eqref{Wkz} is invariant under permutations of the
variables within each group $t_{r_1+\dots+r_{s-1}+1},\dots,t_{r_1+\dots+r_s}$.
Hence every term of the symmetrization occurs exactly $r_1!\cdots r_n!$ times, and
$w_{\vec r}$ is the sum over the cosets $S_r/(S_{r_1}\times\dots\times S_{r_n})$.
Thus formula \eqref{Wkz} makes sense in characteristic $p$ also if some $r_s\ge p$.
See \cite{SV1, FMTV}.

\vsk.2>
Note that the operators $H_m$ in \eqref{kz} differ in sign from those in the
Introduction: in the notation of \eqref{qqKZ} the KZ connection corresponds to
$-H_m=\sum_{j\ne m}\Om^{(m,j)}/(z_m-z_j)$. For example, the classical eigenvalue
relation \eqref{HW} reads $-H_m(z^0)\,w(t^0,z^0)=\psi_m(t^0,z^0)\,w(t^0,z^0)$ in the
present notation.

\subsection{Integral representation}

\begin{thm}
[\cite{SV1}]
\label{thm ri}

There exist $V$-valued rational functions $g_{l,m}(t,z)$, for $l=1,\dots,r$ and $m=1,\dots,n$, such that these functions, together with the rank-one connection $\nabla^o$ defined by \eqref{kzoo}
and the weight function $w$ defined by \eqref{Wkz},
give an integral representation of the KZ connection $\nabla$ defined by \eqref{kz}.
That is,
\bean
\label{Irr}
\left(\nabla^o_{z_m} + \frac 1\ka H_m\right)\!(w)
&=&
\sum_{l=1}^r \nabla^o_{t_l}\left(g_{l,m}\right), \qquad m=1,\dots,n.
\eean

\end{thm}

See formulas for $g_{l,m}$ in \cite[Theorem 2.2]{EV1} and \cite{SV1, SV2}.
\vsk.2>
By Theorem \ref{thm ri}, we can apply Theorem \ref{thm eig} to the connections $\nabla$ and $\nabla^o$
and the weight function $w$.

\vsk.2>

From now on in this section, a point $(t^0,z^0)$ is assumed to satisfy
$t_i^0\ne t_j^0$ for $1\le i<j\le r$, $t_l^0\ne z_m^0$ for all $l,m$, and
$z_a^0\ne z_b^0$ for $1\le a<b\le n$. Such a point is regular, and $w$, $g_{l,m}$,
$h_l$ are regular at it.

\vsk.2>
Given $(t^0,z^0)$, consider the ideal
\bea
\mathcal{J}_p = \langle (t_1-t_1^0)^p, \dots, (t_r-t_r^0)^p, (z_1-z_1^0)^p, \dots, (z_n-z_n^0)^p \rangle \ \subset \ \K[t,z]
\eea
and the quotient $A_{(t^0,z^0)} = \K[t,z]\big/ \mc J_p$.
The connection $\nabla$ induces a connection on $V\ox A_{(t^0,z^0)}$,
and the connection $\nabla^o$ induces a connection on $A_{(t^0,z^0)}$.

\vsk.2>
Denote $\Delta = \prod_{j=1}^r \left(\nabla^o_{t_j}\right)^{p-1}$.
Define an element $I\in V \ox A_{(t^0,z^0)}$
by the formula
\bea
I = \Delta(w).
\eea

\begin{thm}
\label{thm ei}

Let $(t^0, z^0)$ be a solution to the system of Bethe ansatz equations:
\bean
\label{Bea}
\sum_{m=1}^n \frac{-h_p(\La_m)}{t_l^p-z_m^p} + \sum_{j\ne l} \frac {h_p(2)}{t_l^p-t_j^p}\,=\,0
\,, \quad l=1,\dots,r\,.
\eean
Then the element $I \in V\ox A_{(t^0,z^0)}$ is a flat section of the extended connection. That is,
\bea
\left(\der_{z_m}+ \frac 1\ka \psi_m +\frac 1\ka H_m\right) I
&=&
0,\qquad m=1,\dots, n,
\\
\left(\der_{t_l}+\frac 1\ka \phi_l\right) I
&=&
0,\qquad l=1,\dots, r.
\eea
Furthermore, the element $I$ is an eigensection of the $p$-curvature operators $(C_m)$ of the connection $\nabla$
with respective eigenvalues $(-\frac 1{\ka^p}\tilde \psi_m(t^0,z^0))$, that is,
\beq
\label{igv}
C_m\,I \,=\, - \frac 1{\ka^p}
\left( \sum_{l=1}^r \frac {-h_p(\La_m)}{(z_m^0)^p-(t_l^0)^p}+ \sum_{i\ne m}
\frac{h_p(\La_m\La_i/2)}{(z_m^0)^p-(z_i^0)^p} \right) I, \qquad m=1,\dots,n.
\eeq

\end{thm}

\begin{cor}
\label{cor ig}

Let $(t^0, z^0)$ be a solution to the system of Bethe ansatz equations \eqref{Bea}.
Then the vector $I(t^0,z^0) \in V$, if nonzero, is an eigenvector of the $p$-curvature operators $(C_m(z^0))$
with respective eigenvalues $(-\frac1{\ka^p}\tilde \psi_m(t^0,z^0))$. That is,
\bea
C_m(z^0)\,I(t^0,z^0) \,=\, -\, \frac 1{\ka^p}\tilde \psi_m(t^0,z^0)\, I(t^0,z^0), \qquad m=1,\dots,n.
\eea

\end{cor}

See also Lemma \ref{lem dit}.

\begin{rem}
In \cite{EV2}, the spectrum of the $p$-curvature operators of periodic pencils of flat connections is described.
The KZ connection \eqref{kz} is an example of a periodic pencil of flat connections. Formula \eqref{igv} for the eigenvalues of the $p$-curvature operators agrees with the description in \cite{EV2}.

Eigenvectors and eigenvalues of the $p$-curvature operators, as well as the $p$-curvature operators
themselves, are described in \cite{VV1} in the special case when $r=1$ and $\La_m=1$ for $m=1,\dots, n$.

\end{rem}

\subsection{Singular vectors} It is known that the KZ connection on
$\bigotimes^n_{j=1} M_{{\La_j}}$ commutes with the action of $\frak{sl}_2$. Hence
the KZ connection preserves all subspaces of singular vectors
\\
$\Sing \left(\bigotimes^n_{j=1} M_{{\La_j}}\right)\left[r\right]$.
It turns out that the eigenvectors of the $p$-curvature operators provided by Theorem \ref{thm ei} and Corollary \ref{cor ig}
are singular vectors.

\begin{thm}
[\cite{SV1}]
\label{thm sing}
There exist
$\left(\bigotimes^n_{j=1} M_{{\La_j}}\right)\left[r-1\right]$-valued
functions $h_l(t,z)$ for $l=1,\dots,r$, rational in variables $(t,z)$, such that
\bean
\label{Si}
e\, w \,= \,
\sum_{l=1}^r \nabla^o_{t_l}\left(h_{l}\right).
\eean
\end{thm}

\begin{cor}
\label{cor sing}

The  element $I$ of Theorem \ref{thm ei} lies in
\bea
\Sing \left(\bigotimes^n_{j=1} M_{{\La_j}}\right)\left[r\right]\ox A_{(t^0,z^0)}.
\eea

\end{cor}

\begin{proof}
We have
\bea
e\, I
=
\sum_{l=1}^r \left(\nabla^o_{t_l}\right)^p
\left(\prod_{j\ne l}^r \left(\nabla^o_{t_j}\right)^{p-1}\right) \!\!
\left(h_{l}\right)
= \frac 1{\ka^p}\sum_{l=1}^r \tilde \phi_l \cdot \left(\prod_{j\ne l}^r \left(\nabla^o_{t_j}\right)^{p-1}\right) \!\!
\left(h_{l}\right) = 0.
\eea
\end{proof}

\section{Polynomial flat sections of the KZ connection}
\label{sec 4}

To illustrate our results, we consider the special case of the KZ connection \eqref{kz} with fiber $V$
where
\bean
\label{7.2}
\La_{1},\, \dots, \,\La_{n} , \ka \, \in \F_p^\times\subset \K\,.
\eean

\subsection{Master function}

For \,$b\in\F_p$\,, \,let \,$\bar b\in\{0,\ldots, p-1\>\}\subset \Z$
\,be such that \,$\bar b=b\pmod p$\,.
Under assumptions \eqref{7.2}, denote
\be
\Phi(t,z)
=
\prod_{1\le i<j\le n} (z_i-z_j)^{\overline{\La_i\La_j/(2\ka)}}
\prod_{1 \leq i < j \leq r} (t_i-t_j)^{\overline{2/\ka}}
\prod_{m=1}^{n} \prod_{i=1}^{r} (t_i-z_m)^{\overline{-\La_m/\ka}}.
\ee

\begin{lem}
If $f\in\K(t,z)$, then
\bean
\label{ndc}
\Phi \nabla^o_{t_l}(f) = \der_{t_l}(\Phi f),
\qquad
\Phi \nabla^o_{z_m}(f) = \der_{z_m}(\Phi f),
\eean
for all $l,\,m$.
\end{lem}

\begin{proof}
The proof is by inspection.
\end{proof}

\subsection{Bethe ansatz equations}

Under assumptions \eqref{7.2},
we have
$h_p(\La_m)= h_p(2)=h_p(\La_m\La_i/2) =0$
for $1\leq i,m\leq n$, and
\bea
\tilde \phi_l
&=&
 \sum_{m=1}^n \frac{-h_p(\La_m)}{t_l^p-z_m^p} + \sum_{j\ne l} \frac {h_p(2)}{t_l^p-t_j^p} = 0
\,, \quad l=1,\dots,r\,.
\eea
That is, every point $(t^0,z^0)$ is a solution to the system of Bethe ansatz equations \eqref{Bea}.
Furthermore, we have:
\bea
\tilde \psi_m
&=&
 \sum_{l=1}^r \frac {-h_p(\La_m)}{z_m^p-t_l^p}+ \sum_{i\ne m}
\frac{h_p(\La_m\La_i/2)}{z_m^p-z_i^p}=0\,, \quad
m=1,\dots,n\,.
\eea
That is, the functions $(\tilde \psi_{m}(t,z))$ are constants equal to 0.

\subsection{Polynomial flat sections}

The function $\Phi (t,z)\, w(t,z)$ is a $V$-valued polynomial.
Consider the expansion
\bea
\Phi (t,z)\, w(t,z)\,=\!\sum_{i_1,\dots,\>i_r\in\Z_{\geq 0}}\!
Q_{i_1,\dots,\>i_r}(z)\;t_1^{i_1}\!\<\dots t_r^{i_r}\>.
\eea

\begin{thm}
\label{thm Lau}
For any $l_1,\dots,l_r\in \Z_{\geq 0}$, the $V$-valued polynomial
$Q_{pl_1+p-1,\dots,\>pl_r+p-1}(z)$ is a flat section
of the KZ connection \eqref{kz},
\bean
\label{flat lp}
\left(\der_{z_m}
\,+ \,\frac 1\ka \,H_m\right)
Q_{pl_1+p-1,\dots,\>pl_r+p-1}\,=\,0, \quad m=1,\dots,n.
\eean
Furthermore, the polynomial $Q_{pl_1+p-1,\dots,\>pl_r+p-1}(z)$, if nonzero,  is an eigensection of the
$p$-curvature operators with
all eigenvalues equal to 0, that is
\bean
\label{7.4}
C_{m}(z) \,Q_{pl_1+p-1,\dots,\>pl_r+p-1}(z) = 0,
\quad m=1,\dots,n.
\eean
We also have
\bean
\label{7.5}
e \, Q_{pl_1+p-1,\dots,\>pl_r+p-1}(z)= 0.
\eean

\end{thm}

\begin{proof}
By Theorem \ref{thm global}, the function
\bean
\label{JP}
J = \left(\prod_{i=1}^r \der_{t_i}^{p-1} \right)(\Phi \,w) = 
(-1)^r\sum_{l_1,\dots,l_r\geq 0}
Q_{pl_1+p-1,\dots,\>pl_r+p-1}(z)\,t_1^{pl_1}\dots t_r^{pl_r}
\eean
satisfies 
\bea
\left(\der_{z_m}
\,+ \,\frac 1\ka \,H_m\right) J\,=\,0, \quad m=1,\dots,n.
\eea
This implies \eqref{flat lp}. The other claims of the theorem are clear.
\end{proof}

Notice that $\der_{t_l} J=0$ for $l=1,\dots,r$ by Theorem \ref{thm global}. That statement also follows from \eqref{JP}.

\begin{rem}

The statement that the polynomials $Q_{pl_1+p-1,\dots,\>pl_r+p-1}(z)$ are flat sections of the KZ connection was proved in \cite{SV2} by a different construction. The polynomials
\\
$Q_{pl_1+p-1,\dots,\>pl_r+p-1}$ are called the $p$-hypergeometric flat sections.
\end{rem}

\begin{exmp}
Let $r=1$ and $n=2g+1$ for some positive integer $g$. Let $\La_m=1$ for $m=1,\dots,n$. Let $\ka=2$. Then
\bea
\Phi &=& 
\prod_{i<j} (z_i-z_j)^{\overline{1/4}}
\prod_{m=1}^n (t_1-z_m)^{(p-1)/2}\,,
\\
w &=& \sum_{j=1}^n \frac 1{t_1-z_j} v_1\ox \dots \ox fv_1\ox \dots \ox v_1\,,
\eea
where $fv_1$ stays at the $j$th position. Then Theorem \ref{thm Lau} produces $g$
\ $p$-hypergeometric flat sections $Q_{p-1}$, $Q_{2p-1}$, \dots, $Q_{gp-1}$ for $p > 2g+1$.

\vsk.2>

It is shown in \cite{VV1} that these solutions span all rational flat sections of the KZ connection in this case, cf. \cite{MV2}.  See also \cite{VV2}.

\end{exmp}

\section{KZ and dynamical connections}
\label{sec 5}

\subsection{Compatible KZ and dynamical connections}

Fix a degree subspace
\bea
\tsize V=\left(\bigotimes^n_{j=1} M_{{\La_j}}\right)\left[r\right]\,.
\eea
Let $z=(z_1,\dots,z_n)$, $\la$ be variables. Define
the Gaudin Hamiltonians and the dynamical Hamiltonian
by the formulas:
\bea
H_m(z_1,\dots,z_{n}, \la)
&=&
-\frac \la 2\, h^{(m)}\, -\sum_{j\ne m}\frac{\Om^{(m,j)}}{z_m-z_j}\,,
\quad m=1,\dots,n,
\\
D(z_1,\dots,z_{n}, \la)
&=&
-\sum_{i=1}^n \, \frac{z_i}2\,h^{(i)}\, - \sum_{i,j=1}^n \frac{f^{(i)}e^{(j)}}\la\,.
\eea
These are elements of $(\End V)(z,\la)$. Fix
\bea
\ka\,\in\,\K^\times.
\eea
Define a connection $\nabla$ on the trivial bundle with base $\A^{n+1}$ and fiber $V$
by the formulas:
\bean
\label{kzd}
\nabla_m
&=&
\der_{z_m} + \frac 1\ka H_m\,,\qquad m=1,\dots,n,
\\
\notag
\nabla_{n+1}
&=&
\der_{\la} + \frac 1\ka D\,.
\eean
The connection is flat by \cite{FMTV}. The system  $\nabla$ is called the {\it system of 
compatible KZ and dynamical connections}.
The connection $\nabla$ has singularities at the poles of the operators $(H_m, D)$.
Denote by $C_m=(\nabla_m)^p$, $m=1,\dots,n+1$, the $p$-curvature operators of $\nabla$.

\subsection{Auxiliary flat connection of rank one}

Let $\A^{r+n+1}$ be the affine space over $\K$ of dimension $r+n+1$ with coordinates
$(t, z,\la)$.
Consider the rank-one connection $\nabla^o$ with base $\A^{r+n+1}$ defined by the
formulas:
\bea
\nabla^o_{t_l}
&=&
\der_{t_l} +\frac1\ka \phi_l, \qquad l=1,\dots, r,
\\
\nabla^o_{z_m}
&=&
\der_{z_m}+ \frac1\ka \psi_m, \qquad m=1,\dots, n,
\\
\nabla^o_{\la}
&=&
\der_{\la}+\frac1\ka \psi_{n+1},
\eea
where
\bean
\label{kzdo}
\phi_l
&=&
-\la + \sum_{m=1}^n \frac{-\La_m}{t_l-z_m} + \sum_{j\ne l} \frac 2{t_l-t_j}\,, \quad l=1,\dots,r\,,
\\
\notag
\psi_m
&=&
\frac {\la\La_m} 2 + \sum_{l=1}^r \frac {-\La_m}{z_m-t_l}+ \sum_{i\ne m} \frac{\La_m\La_i/2}{z_m-z_i} \,, \quad
m=1,\dots,n\,,
\\
\notag
\psi_{n+1}
&=&
\frac 1 2\, \sum_{m=1}^n\La_m z_m - \sum_{l=1}^r t_l\,.
\eean
The connection $\nabla^o$ has singularities at the poles of the functions $(\phi_l,\psi_m,\psi_{n+1})$.

\begin{lem}
The connection $\nabla^o$ is flat.
\end{lem}

\begin{proof}

The proof is by inspection.
\end{proof}

\begin{lem}
The $p$-curvature operators of $\nabla^o$ are given by the formulas:
\bea
(\nabla_{t_l}^o)^p
&=&
\frac 1{\ka^p}\tilde \phi_l,\qquad l=1,\dots,r,
\\
(\nabla_{z_m}^o)^p
&=&
\frac 1{\ka^p}\tilde \psi_m, \qquad m=1,\dots,n,
\\
(\nabla_{\la}^o)^p
&=&
\frac 1{\ka^p}\tilde \psi_{n+1}\,,
\eea
where
\bea
\tilde \phi_l
&=&
-\la^p + \sum_{m=1}^n \frac{-h_p(\La_m)}{t_l^p-z_m^p} + \sum_{j\ne l} \frac {h_p(2)}{t_l^p-t_j^p}
\,, \quad l=1,\dots,r\,,
\\
\tilde \psi_m
&=&
 \frac {\la^p\La_m^p} 2 + \sum_{l=1}^r \frac {-h_p(\La_m)}{z_m^p-t_l^p}+ \sum_{i\ne m} \frac{h_p(\La_m\La_i/2)}{z_m^p-z_i^p}\,, \quad
m=1,\dots,n\,,
\\
\tilde \psi_{n+1}
&=&
 \frac 1 2\, \sum_{m=1}^n\La_m^pz_m^p - \sum_{l=1}^r t_l^p\,.
\eea

\end{lem}

\subsection{Weight function}

Define the $V$-valued weight function
\bea
w(t,z)\,=\>\sum_{\vec r\in \mc I_r} w_{\vec r}(t;z)\,f^{\<\>\vec r}
\eea
by formula \eqref{Wkz}.

\subsection{Integral representation}

\begin{thm}
[\cite{FMTV}]
\label{thm ir}

There exist $V$-valued rational functions $g_{l,m}(t,z,\lambda)$, for $l=1,\dots,r$ and $m=1,\dots,n+1$, such that these functions, together with the rank-one connection $\nabla^o$ defined by \eqref{kzdo}
and the weight function $w$ defined by \eqref{Wkz},
give an integral representation of the compatible KZ and dynamical connections
$\nabla$ defined by \eqref{kzd}.
That is,
\bean
\label{IRR}
\left(\nabla^o_{z_m} +\frac 1\ka H_m\right)\!(w)
&=&
\sum_{l=1}^r \nabla^o_{t_l}\left(g_{l,m}\right), \qquad m=1,\dots,n,
\\
\notag
\left(\nabla^o_{\la} + \frac 1\ka D\right)\!(w)
&=&
\sum_{l=1}^r \nabla^o_{t_l}\left(g_{l,n+1}\right).
\eean

\end{thm}

See formulas for $g_{l,m}$ in \cite[Theorem 2.2]{EV1} and \cite{SV1, FMTV}.
\vsk.2>
By Theorem \ref{thm ir}, we can apply Theorem \ref{thm eig} to the connections $\nabla$ and $\nabla^o$
and the weight function $w$.

From now on in this section, a point $(t^0,z^0,\la^0)$ is assumed to satisfy
$\la^0\ne0$, $t_i^0\ne t_j^0$ for $1\le i<j\le r$, $t_l^0\ne z_m^0$ for all $l,m$,
and $z_a^0\ne z_b^0$ for $1\le a<b\le n$. Such a point is regular, and $w$ and
$g_{l,m}$ are regular at it.

\vsk.2>
Given $(t^0,z^0, \la^0)$, consider the ideal
\bea
\mathcal{J}_p = \langle (t_1-t_1^0)^p, \dots, (t_r-t_r^0)^p, (z_1-z_1^0)^p, \dots, (z_n-z_n^0)^p, (\la-\la^0)^p \rangle \ \subset \ \K[t,z,\la]
\eea
and the quotient $A_{(t^0,z^0,\la^0)} = \K[t,z,\la]\big/ \mc J_p$.
The connection $\nabla$ induces a connection on $V\ox A_{(t^0,z^0,\la^0)}$,
and the connection $\nabla^o$ induces a connection on $A_{(t^0,z^0,\la^0)}$.
We keep the notation $\nabla$ and $\nabla^o$ for the induced connections.

\vsk.2>
Let $\Delta = \prod_{j=1}^r \left(\nabla^o_{t_j}\right)^{p-1}$.
Define an element $I\in V \ox A_{(t^0,z^0, \la^0)}$
by the formula
\bea
I = \Delta(w).
\eea

\begin{thm}
\label{thm kzd}

Let $(t^0, z^0, \la^0)$ be a solution to the system of Bethe ansatz equations:
\bean
\label{bbea}
-\la^p + \sum_{m=1}^n \frac{-h_p(\La_m)}{t_l^p-z_m^p} + \sum_{j\ne l} \frac {h_p(2)}{t_l^p-t_j^p}\,=\,0
\,, \quad l=1,\dots,r\,.
\eean
Then the element $I \in V\ox A_{(t^0,z^0,\la^0)}$ is a flat section of the extended connection. That is,
\bea
\left(\der_{z_m}+ \frac 1\ka\psi_m +\frac1\ka H_m\right) I
&=&
0,\qquad m=1,\dots, n,
\\
\left(\der_{\la}+ \frac 1\ka \psi_{n+1} +\frac 1\ka D\right) I
&=&
0,
\\
\left(\der_{t_l}+ \frac 1\ka\phi_l\right) I
&=&
0,\qquad l=1,\dots, r.
\eea
Furthermore, the element $I$ is an eigensection of the $p$-curvature operators $(C_m)$
of the connection $\nabla$
with respective eigenvalues $(-\frac 1{\ka^p} \tilde \psi_m(t^0,z^0,\la^0))$, $m=1,\dots,n+1$. That is,
\bean
\label{igvV}
C_m\,I \,
&=&\, - \frac 1{\ka^p}
\left( \frac {\left(\la^0\right)^p\La_m^p} 2 + \sum_{l=1}^r \frac {-h_p(\La_m)}{(z_m^0)^p-(t_l^0)^p}+ \sum_{i\ne m}
\frac{h_p(\La_m\La_i/2)}{(z_m^0)^p-(z_i^0)^p} \right) I, 
\eean
where $m=1,\dots,n$,
\bean
\label{Dc}
\phantom{aaa}
C_{n+1}\,I \,&=&\,
- \frac 1{\ka^p} \left(\frac 1 2\, \sum_{m=1}^n\La_m^p(z_m^0)^p - \sum_{l=1}^r (t_l^0)^p\right) I\,.
\eean

\end{thm}

\begin{cor}
\label{cor last}

Let $(t^0, z^0, \la^0)$ be a solution to the system of Bethe ansatz equations \eqref{bbea}.
Then the vector $I(t^0,z^0,\la^0) \in V$, if nonzero, is an eigenvector of the $p$-curvature operators $(C_m(z^0,\la^0))$
of the connection $\nabla$
with respective eigenvalues $(-\frac 1{\ka^p}\tilde \psi_m(t^0,z^0,\la^0))$, $m=1,\dots,n+1$. That is,
\bean
\label{igv l}
&&
C_m(z^0,\la^0)\,I(t^0,z^0,\la^0) \,
\\
\notag
&&
= \, - \frac 1{\ka^p}
\left( \frac {\left(\la^0\right)^p\La_m^p} 2 + \sum_{l=1}^r \frac {-h_p(\La_m)}{(z_m^0)^p-(t_l^0)^p}
+ \sum_{i\ne m}
\frac{h_p(\La_m\La_i/2)}{(z_m^0)^p-(z_i^0)^p} \right) I(t^0,z^0,\la^0), 
\eean
where $m=1,\dots,n$,
\bean
\label{Dc l}
\phantom{aaa}
&&
C_{n+1}(z^0,\la^0)\,I(t^0,z^0,\la^0)
 \,=\,
- \frac 1{\ka^p} \left(\frac 1 2\, \sum_{m=1}^n\La_m^p (z_m^0)^p - \sum_{l=1}^r (t_l^0)^p\right) I(t^0,z^0,\la^0)\,.
\eean

\end{cor}

\end{document}